\documentclass[conference]{IEEEtran}
\usepackage{fancyhdr}
\usepackage{float}
\usepackage{graphicx}
\usepackage{epsfig}
\usepackage{bm}
\usepackage{stfloats}
\usepackage{amsmath}
\usepackage{mathrsfs}
\usepackage{subfigure}
\usepackage{array}
\usepackage{booktabs}
\usepackage{amssymb}
\usepackage{graphicx}
\usepackage{epstopdf}
\usepackage{algorithm}
\usepackage{amsthm}
\usepackage{bbm}
\usepackage{graphicx}
\usepackage{epsfig}
\usepackage{color}

\newtheorem{theorem}{Theorem}

\newtheorem{cor}{Corollary}
\makeatother

\IEEEoverridecommandlockouts
\author{\IEEEauthorblockN{Chang Liu$^\star$, Ming Ding$^\dag$, Chuan Ma$^\star$, Qingzhi Li$^\star$, Zihuai Lin$^\star$ and Ying-Chang Liang$^{\ddagger}$}\\
{$^\star$ School of Electrical and Information Engineering, University of Sydney, Sydney, NSW, Australia}\\
{$^\dag$ Data61, CSIRO, Australia}\\
{$^\ddagger$ University of Electronic Science and Technology of China (UESTC), Chengdu, China}\\
\small{Email: \{cliu6845, qili3928\}@uni.sydney.edu.au, \{chuan.ma, zihuai.lin\}@sydney.edu.au, Ming.Ding@data61.csicro.au, liangyc@ieee.org}
\thanks{The work of Y.-C. Liang is funded by National Natural Science Foundation of China under Grants 61571100, 61631005 and 61628103.}}
\begin{document}

\title{Performance Analysis for Practical Unmanned Aerial Vehicle Networks with LoS/NLoS Transmissions}
\maketitle
\begin{abstract}
  In this paper, we provide a performance analysis for practical unmanned aerial vehicle (UAV)-enabled networks. By considering both line-of-sight (LoS) and non-line-of-sight (NLoS) transmissions between aerial base stations (BSs) and ground users, the coverage probability and the area spectral efficiency (ASE) are derived.
  Considering that there is no consensus on the path loss model for studying UAVs in the literature, in this paper,
  three path loss models, i.e., high-altitude model, low-altitude model and ultra-low-altitude model, are investigated and compared.
  Moreover,
  the lower bound of the network performance is obtained assuming that UAVs are hovering randomly according to homogeneous Poisson point process (HPPP), while the upper bound is derived assuming that UAVs can instantaneously move to the positions directly overhead ground users.
  From our analytical and simulation results for a practical UAV height of 50 meters, we find that the network performance of the high-altitude model and the low-altitude model exhibit similar trends, while that of the ultra-low-altitude model deviates significantly from the above two models.
  In addition, the optimal density of UAVs to maximize the coverage probability performance has also been investigated.
  \end{abstract}
\section{Introduction}
 Due to the flying nature of unmanned aerial vehicles (UAVs), base stations (BSs) can be mounted on the UAV to support wireless communications and improve the performance of cellular networks. For example, UAV-mounted base stations (UAV-BSs) are introduced when a natural disaster interrupts communications or ground base stations are overloaded \cite{zhang2017spectrum}. Compared with ground BSs, the flexibility of UAV-BSs allows them to adapt their locations to the demand of users.

Most of the literature on the UAV-BS focuses on its deployment. The work in \cite{ono2016wireless} proposed that fixed-wing UAVs at a constant height are more applicable for aerial networks due to less power consumption.
Positions of UAV-BSs were modeled as a 3D Poisson Point Process (3D-PPP) distribution with a limited height in \cite{zhang2017spectrum}, but the analysis in \cite{gruber2016role} showed that the flexible height of UAV is not as helpful as a well-chosen fixed altitude. In \cite{galkin2016deployment}, UAV-mounted mobile base stations were deployed in a fixed altitude and placed along an optimal trajectory to cover as much as user equipment (UE) whose locations are already known in a given area.
%In \cite{galkin2016deployment}, the ground terminals are divided into subsets by K-Means based on the locations of ground users and mobile base station stations are placed over the center of each subsets.

Beyond the UAV deployment, the performance of 3D networks also attracts much attention in the existing literature. The work in \cite{gruber2016role} analyzed the average downlink spectral efficiency without considering the environment noise, while the authors of \cite{mozaffari2015drone} evaluated the performance of UAV at a low altitude platform in terms of the coverage area and transmit power. %In \cite{bor2016efficient}, the maximum number of covered users is obtained in finding optimal deployment.
Similarly, the optimal deployment model in \cite{alzenad20173d} led to the analysis of coverage and transmit power. Furthermore, the analysis in \cite{mozaffari2016unmanned} introduced a tractable analytical framework for the coverage and the rate in UAV based network with the coexistence of device-to-device (D2D) network.

Although the path loss model has been considered as a key factor in the performance analysis for UAV networks, there is no consensus on this issue yet. For example, the work in \cite{zhang2017spectrum} and \cite{galkin2016deployment} only considered the UAV hovering in a LoS dominated network for simplicity. To conduct a practical analysis for UAV, the authors of \cite{bor2016efficient} proposed a general path loss model which considers both LoS and NLoS connections and their occurrence probabilities, depending on the elevation angle between a UAV and a user. Despite that this model has been widely adopted as the high-altitude model (a typical height is around 1000 meters), the network performance has not been investigated due to the complexity of the proposed model. On the other hand, the work in \cite{ding2016please} provided a network analysis of the terrestrial cellular network where the antenna height between BSs and users is around 10m$\sim$30m, together with 3GPP LoS and NLoS models. Considering that the height of UAVs is comparable with that of ground base stations in future UAV networks, the curent macrocell-to-UE model (a typical height is around 32 meters) and picocell-to-UE model (a typical height is around 10 meters) proposed for terrestrial communication in 3GPP standard can also be applied to the UAV-based network. Such macrocell-to-UE model and picocell-to-UE model are referred to as the low-altitude model and the ultra-low-altitude model hereafter. To our best knowledge, the path loss model for UAV-BSs has not been adequately explored in the literature, so the intriguing question arises: which is the most appropriate path loss model when UAVs fly at a practical height, e.g., 50m$\sim$100m? The reasons for such medium-altitude deployment of UAVs are: i) UAVs should not fly too high (e.g., larger than 100 meters) because of the recently discovered network capacity crash \cite{ding2016performance}, and ii) UAVs should not fly too low (e.g., lower than 10 meters) due to obvious safety reasons.
%the performance analysis presented in \cite{zhang2017spectrum}, \cite{lyu2017placement} and \cite{galkin2016deployment} are only applicable to the LoS dominated network, and the work in \cite{mozaffari2015drone} and \cite{bor2016efficient}, which considers the LoS and NLoS connection, does not give the performance analysis. To our best knowledge, the theoretical study of UAV-BS networks with considerable LoS and NLoS connection has not been conducted before.
%The authors in \cite{ding2016performance} analyze two special cases: 3GPP case 1 and 3GPP case 2, and then introduce the piece-wise probability function to fit general path loss model, and finally give the close-form solution of the coverage probability and the ASE. In \cite{andrews2011tractable}, the author introduces that if the distribution of base station follows Homogenous Poisson point process (HPPP), the model can provide a reliable lower bound for performance.

Motivated by the above theoretical gap and to answer such fundamental question, in this paper, we analyze the performance of UAV-enabled networks on the condition of different path loss models. An interesting finding in our study is that although the 3GPP path loss model is developed for terrestrial communications, the network performance based on the low-altitude model (a typical height is around 32 meters) is similar to that based on the high-altitude model, when a practical UAV height of 50 meters is considered. Moreover, the optimal UAV density to maximize the coverage probability and the area spectral efficiency (ASE) performance can be found from numerical results.

The rest of this paper is structured as follows. In Section II we introduce the system model of UAV-enabled networks. In section III, three path loss models are investigated and compared. In section IV, we derive the analytical expressions of the coverage probability and the ASE for a 3D network with PPP-distributed UAVs, and then we propose an ideal case with teleportation UAVs to find the upper bound of the network performance. Numerical results and discussions are provided in Section V and the conclusions are drawn in Section VI.

\section{System Model}
We consider a UAV network, where UAV aerial base stations follow a 3D-PPP distribution with a density $\lambda$ in an infinite 3D space $\mathbb{V}$, and the UAV height is set to $h$, that is $\mathbb{V}=\{(x,y,z):x,y\in \mathbb{R}, z=h \}$. Here, we consider practical values for $h$ around 50$\sim$100 meters. We consider such medium-altitude deployment of UAVs because UAVs should not fly too high (e.g., larger than 100 meters) since the recently discovered network capacity crash \cite{ding2016please}, and UAVs should not fly too low (e.g., lower than 10 meters) due to obvious safety reasons. User equipments (UEs) are Poisson distributed in the considered network with a density of $\lambda_{\textrm{UE}}$. Here, $\lambda_{\textrm{UE}}$ is assumed to be sufficiently larger than $\lambda$ so that each UAV has at least one associated UE in its coverage area. The 3D distance between an arbitrary UAV and an arbitary UE is denoted by $r$ in km. Considering practical LoS and NLoS transmissions, we propose to model the path loss associated with distance $r$ as a path loss function $\zeta(r)$.
Such $\zeta(r)$ is segmented into 2 pieces, where $\zeta^{\textrm{L}}(r)$ is the path loss function for LoS transmission, $\zeta^{\textrm{NL}}(r)$ is the path loss function for NLoS transmission and $\Pr^{\textrm{L}}(r)$ is the LoS probability function. In more detail,
\begin{itemize}
  \item $\zeta(r)$ is modeled as
  \begin{equation}\label{pathloss}
{\zeta}(r) = \left\{ {\begin{array}{*{20}{l}}
{\zeta^{\textrm{L}}(r) = A^{\textrm{L}}{r^{ - \alpha ^{\textrm{L}}}},~~~~~~\text{for LoS}}\\
{\zeta^{\textrm{NL}}(r) = A^{\textrm{NL}}{r^{ - \alpha ^{\textrm{NL}}}},~\text{for NLoS}}
\end{array}} \right.,
  \end{equation}
  with $A^{\textrm{L}}$ and $A^{\textrm{NL}}$ being the path losses at a reference distance $r=1$ and $\alpha^{\textrm{L}}$ and $\alpha^{\textrm{NL}}$ being the path loss exponents for the LoS and the NLoS cases in $\zeta(r)$, respectively. In practice, $A^{\textrm{L}}$, $A^{\textrm{NL}}$, $\alpha^{\textrm{L}}$ and $\alpha^{\textrm{NL}}$ are constants obtained from field tests \cite{ding2016performance}.
  \item $\Pr^{\textrm{L}} (r)$ is the probability function that a transmitter and a receiver have LoS connections. Also, the probability of NLoS is $\Pr^{\textrm{NL}} (r)=1-\Pr^{\textrm{L}} (r)$.
\end{itemize}
As a common practice in the field, each UE is assumed to be associated with the UAV that provides the strongest signal strength, and the multi-path fading between an arbitrary UAV-BS and an arbitrary UE is modeled as independently identical distributed (i.i.d.) Rayleigh fading. Thus, the channel gain denoted by $g$ can be modeled as an i.i.d. exponential random variable (RV). The transmit power of each UAV and the additive white Gaussian noise (AWGN) power at each UE are denoted by $P$ and $\sigma^2$, respectively.
%\begin{figure}
%  \centering
%  % Requires \usepackage{graphicx}
%  \includegraphics[width=3.5in]{UAVFIG.jpg}\\
%  \caption{System model}
%  \label{system model}
%\end{figure}
%
%As shown in Fig. 1, except the associated UAV-BS, other flying UAVs are regarded as interfering nodes.
\section{Discussion and Analysis of path loss models}
Since there is no consensus on proper path loss model for UAV-enabled networks, we choose three widely adopted path loss models and apply them to the considered UAV networks.

\subsection{High-altitude model}
The high-altitude model based on the elevation angle has been widely used in the satellite communication model, e.g., thousands of meters. The probability function that a transmitter and a receiver have a LoS connection at an elevation angle of $\theta$ can be expressed as $\Pr^{\textrm{L}}(\theta)$ \cite{mozaffari2016unmanned}:
\begin{equation}
 \begin{split}
 {{\Pr} ^{\textrm{L}}}(\theta ) = \frac{1}{{1 + C \exp \left( { - B [\theta  - C ]} \right)}},
 \end{split}
\end{equation}
where $B$ and $C$ are constant values that depend on the environment (rural, urban, dense urban, etc.). Furthermore, the elevation angle $\theta$ can be written as $\theta {\rm{ = }}\frac{{180}}{\pi } \arcsin \left( {\frac{h}{r}} \right)$, so the LoS probability function for this high-altitude model can be reformulated as a new function with respect to $r$:
\begin{equation}
 \begin{split}
 {{\Pr}_{\textrm{high}} ^{\textrm{L}}}(r ) = \frac{1}{{1 + C \exp \left( { - B [\frac{{180}}{\pi }\arcsin \left( {\frac{h}{r}} \right) - C ]} \right)}}.
 \end{split}
\end{equation}

%According to \cite{fotouhi2017dynamic}, the path losses at reference distance $A^\textrm{L}$ and $A^{\textrm{NL}}$, and the path loss exponents for LoS and NLoS cases ${\alpha}^\textrm{L}$ and ${\alpha}^{\textrm{NL}}$ are set as the same as the values of pico-UE path loss model.
\subsection{Low-altitude model }
 Provided that the practical height of UAV-BSs is usually limited to a medium altitude, like 50m and 100m.
 %Here, we assume that UAVs fly at the medium altitude because: i) UAVs should not fly too high (e.g., larger than 100 meters) because of the recently discovered network capacity crash [10], and ii) UAVs should not fly too low (e.g., lower than 10 meters) due to obvious safety reasons.
 Such height is comparable to the antenna height of terrestrial base stations, we further analyze the path loss model proposed for 3GPP terrestrial communications and apply it to the considered UAV networks.

 In particular, the 3GPP macrocell-to-UE path loss model has been proposed for connection between a UE and its associated macrocell BS. Considering that the height of a macrocell base station is usually around 32m, which is slightly lower than the considered altitude of UAV around 50$\sim$100m, it is reasonable to use this model to study the UAV network. In this case, the LoS probability function for this low-altitude model can be expressed as \cite{3GPP2012TR36.828}
\begin{equation}
\begin{split}
{{\Pr}_{\textrm{low}} ^\textrm{L}} (r) = &\min \left( {0.018/r,1} \right) \times \left( {1 - \exp \left( { - r/0.063} \right)} \right) \\
&+ \exp \left( { - r/0.063} \right).
\end{split}
\end{equation}

\subsection{Ultra-low-altitude model}
To obtain a comprehensive insight of the proper path loss model for UAVs, we also introduce the 3GPP picocell-to-UE model as the ultra-low-altitude model, since the typical height of picocell base station is about 10m. In this case, the LoS probability function is defined as \cite{3GPP2012TR36.828}
\begin{equation}
\begin{split}
{{\Pr}_{\textrm{ultra}} ^\textrm{L}} (r)= &0.5 - \min \left( {0.5,5\exp \left( { - 0.156/r} \right)} \right) \\
&+ \min \left( {0.5,5\exp \left( { - r/0.03} \right)} \right).
\end{split}
\end{equation}

\subsection{The Comparison of the Three Path Loss Models}
%It should be noted that the path loss parameters, including path loss at a reference distance, and the path loss exponents for the LoS connection and the NLoS connection, are different for pico-UE model and macro-UE model.
\begin{figure}
  \centering
  \includegraphics[width=2.8in]{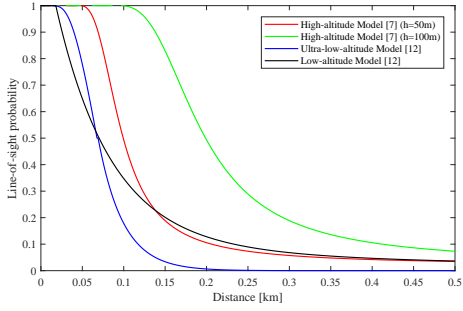}\\
  \caption{Comparison of LoS probability function}
  \label{compare_prob}
\end{figure}

Fig. \ref{compare_prob} compares the LoS probability functions for different path loss models. It can be seen from this figure that the LoS probability for the ultra-low-altitude model drops very quickly with respect to the distance, followed by the low-altitude model. Moreover, it should be noted that the high-altitude model generates different LoS probability functions for different altitudes.

\section{Analysis for the Proposed UAV Networks}
To analyze the performance of UAV-BSs based on the interested path loss models, we investigate the coverage probability and the ASE of the network in this section.
The coverage probability represents the probability that the typical user is covered by the associated UAV-BS and defined as the probability that the received signal-to-interference-noise-ratio (SINR) is larger than a pre-set threshold $\gamma$, which can be expressed as
\begin{equation}
 \begin{split}
 {p^{{\mathop{\rm cov}} }} = \Pr (\textrm{SINR} > \gamma ),
 \end{split}
\end{equation}
where SINR is expressed as
\begin{equation}
 \begin{split}
 \textrm{SINR} = \frac{{Pg\zeta (r)}}{{{I_r} + {N_0}}},
 \end{split}
\end{equation}
where $P$ and ${N_0}$ denote the transmission power of the UAV-BS and additive white Gauss noise (AWGN) power, respectively. Moreover, ${I_r}$ is the sum of interference from other UAV-BSs, and $g$ is the channel gain of Rayleigh fading and can be modeled as a RV which follows a exponential distribution with the mean value of one. It can be further written as
\begin{equation}
 \begin{split}
{I_r} = \sum\limits_{i:{b_i} \in \Phi \backslash {b_o}} {P{\beta _i}{g_{i}}}.
 \end{split}
\end{equation}
Obviously, when UAV-BSs are HPPP distributed and randomly hovering in the network, the network performance reaches a lower bound because the mobility of UAVs is completed ignored. Such lower-bound performance is characterized in the following Theorem 1.
\begin{theorem}\label{1}
 Considering the path loss of the LoS and the NLoS connection, the lower bound of the coverage probability ${p^{{\mathop{\rm cov}} }}(\lambda ,\gamma )$ can be expressed as
\begin{equation}
 \begin{split}
 {p_{\rm{lower}}^{{\rm{cov}}}}(\lambda ,\gamma ) = {T^\emph{{L}}} + {T^{\emph{{NL}}}},
 \end{split}
\end{equation}
where ${T^{\rm{L}}} = \int_h^\infty  {\Pr \left[ {\frac{{P\zeta ^{\rm{L}}(r)g}}{{{I_r} + {N_0}}} > \gamma } \right]} {f^{\rm{L}}}(r)dr$ and ${T^{\rm{NL}}} = \int_h^\infty  {\Pr \left[ {\frac{{P\zeta ^{\rm{NL}}(r)g}}{{{I_r} + {N_0}}} > \gamma } \right]} {f^{\rm{NL}}}(r)dr$.\\
The $f^{\emph{L}}(r)$ and $f^{\emph{NL}}(r)$ are expressed as
\begin{equation}
 \begin{split}
f^{\emph{L}}(r) = &\exp \left( { - \int_h^{{r_1}} {\left( {1 - {{\Pr }^{\emph{L}}}\left( {u} \right)} \right)2\pi u\lambda du} } \right)\\&\times \exp \left( { - \int_h^r {{{\Pr }^{\emph{L}}}\left( {u} \right)2\pi u\lambda du} } \right)\\&\times {\Pr} ^{\emph{L}}\left( {r} \right) \times 2\pi r\lambda,
 \end{split}
\end{equation}
and
\begin{equation}
 \begin{split}
f^{{\emph{NL}}}(r) = &\exp \left( { - \int_h^{{r_2}} {{{\Pr }^{\emph{L}}}\left( {u} \right)2\pi u\lambda du} } \right)\\&\times \exp \left( { - \int_h^r {\left( {1 - {{\Pr }^\emph{L}}\left( {u} \right)} \right)2\pi u\lambda du} } \right)\\&\times \left( {1 - {\Pr}^\emph{L}\left( {r} \right)} \right) \times 2\pi r\lambda,
 \end{split}
\end{equation}
where ${r_1}$ and ${r_2}$ are the solutions of $\zeta ^{\emph{NL}}({r_1}) = \zeta ^\emph{L}(r)$ and ${\zeta ^\emph{L}}({r_2}) = \zeta ^{\emph{NL}}(r)$, respectively.\\
Moreover, $\Pr \left[ {\frac{{P{\zeta ^\emph{L}}(r)g}}{{{I_r} + {N_0}}} > \gamma } \right]$ and ${\Pr \left[ {\frac{{P\zeta ^{\emph{NL}}(r)g}}{{{I_r} + {N_0}}} > \gamma } \right]}$ are expressed by
\begin{equation}
 \begin{split}
 {\Pr \left[ {\frac{{P\zeta ^\emph{L}(r)g}}{{{I_r} + {N_0}}} > \gamma } \right]} = \exp \left( { - \frac{{\gamma {N_0}}}{{P\zeta ^\emph{L}\left( r \right)}}} \right){{\mathcal{L}}_{{I_r}}}\left( {\frac{\gamma }{{P\zeta ^\emph{L}\left( r \right)}}} \right),
\end{split}
\end{equation}
and
\begin{equation}
 \begin{split}
 {\Pr \left[ {\frac{{P\zeta ^{\emph{NL}}(r)g}}{{{I_r} + {N_0}}} > \gamma } \right]} = \exp \left( { - \frac{{\gamma {N_0}}}{{P\zeta ^{\emph{NL}}\left( r \right)}}} \right){{\cal L}_{{I_r}}}\left( {\frac{\gamma }{{P\zeta ^{\emph{NL}}\left( r \right)}}} \right),
\end{split}
\end{equation}
where ${\cal L}_{{I_r}}$ is the Laplace transform of ${I_r}$ in the computation of interference.
\end{theorem}
\begin{proof}
%The proof is omitted due to the page limitation.
See Appendix A.
\end{proof}

In Theorem 1, we assume UAVs are randomly hovering in the network. On the other hand, if we consider the mobility of UAVs, the system performance can surely be improved. However, the analysis of such mobile UAVs is difficult because we need to further consider UAV mobility control management. Fortunately, we can instead consider a UAV teleportation model, where UAVs can instantaneously move to the positions directly overhead the users to show the upper-bound performance of a UAV network. In this case, each user will be associated with its UAV-BS overhead. Such upper-bound performance is characterized in the following Corollary 1, which is derived from Theorem 1.

\begin{cor}
The coverage probability of teleporting UAVs can be expressed as
\begin{equation}
 \begin{split}
 p_{\rm{upper}}^{{\mathop{\rm cov}} }(\lambda ,\gamma )
 = &\Pr \left[ {\frac{{P\zeta ^{\rm{L}}(h)g}}{{{I_r} + {N_0}}} > \gamma } \right] + \Pr \left[ {\frac{{P\zeta ^{\rm{NL}}(h)g}}{{{I_r} + {N_0}}} > \gamma } \right]\\
 = &\exp \left( { - \frac{{\gamma {N_0}}}{{P\zeta ^\emph{L}\left( h \right)}}} \right){{\mathcal{L}}_{{I_r}}}\left( {\frac{\gamma }{{P\zeta ^\emph{L}\left( h \right)}}} \right) \\
 &+ \exp \left( { - \frac{{\gamma {N_0}}}{{P\zeta ^{\emph{NL}}\left( h \right)}}} \right){{\cal L}_{{I_r}}}\left( {\frac{\gamma }{{P\zeta ^{\emph{NL}}\left( h \right)}}} \right).
 \end{split}
\end{equation}
\end{cor}
In this case,  the associated UAV is set at the positions overhead the users, so the space distance from users to their associated UAVs is $h$ rather than $r$.
In comparison with the case of HPPP distributed UAVs, the case of teleporting UAVs can provide the user with the strongest received signal power due to the minimized distance between them and the highest probability of having LoS connection. As a result, this teleporting model gives the upper bound of network performance.
The detailed discussion on such upper-bound performance will be shown in Sec.V-C.

According to \cite{ding2017new}, the ASE can be expressed as
\begin{equation}
 \begin{split}
 {A^{{\rm{ASE}}}}\left( {\lambda ,{\gamma _{\rm{0}}}} \right) = &\frac{\lambda }{{\ln 2}}\int_{{\gamma _0}}^{ + \infty } {\frac{{{p^{{\rm{cov}}}}(\lambda ,\gamma )}}{{1 + \gamma }}} d\gamma  \\&+ \lambda {\log _2}\left( {1 + {\gamma _{\rm{0}}}} \right){p^{{\rm{cov}}}}(\lambda ,{\gamma _0}),
 \end{split}
\end{equation}
where ${\gamma _0}$ is the minimum SINR threshold for UE to work normally.

\section{Simulation Results}
To find the appropriate path loss model when UAVs fly at a medium altitude, we use simulation results to demonstrate the coverage probability and the ASE of three LoS probability models and make a comparison. Parameters adopted in simulation are: $P=24$ dBm, $N_0=-95$ dBm \cite{3GPP2012TR36.828}, $\gamma_0=0$ dB, $C=11.95$, $B=0.136$ \cite{mozaffari2016unmanned}. To obtain the numerical results at the medium height, we choose to analyze UAVs at the height of 50m and 100m, which are the most practical cases in reality. For the high-altitude model, the relative parameters are: $A^\textrm{L}=10.38$, $A^\textrm{NL}=14.54$, $\alpha^\textrm{L}=2.09$, $\alpha^\textrm{NL}=3.75$ \cite{galkin2016deployment}, \cite{fotouhi2017dynamic}.
%$A^\textrm{L}=9.95$, $A^\textrm{NL}=11.85$, $\alpha^\textrm{L}=2$, $\alpha^\textrm{NL}=2$.
For the low-altitude model, path loss parameters are: $A^\textrm{L}=10.34$, $A^\textrm{NL}=13.11$, $\alpha^\textrm{L}=2.42$, $\alpha^\textrm{NL}=4.28$ \cite{3GPP2012TR36.828}. For the ultra-low-altitude model, path loss parameters are: $A^\textrm{L}=10.38$, $A^\textrm{NL}=14.54$, $\alpha^\textrm{L}=2.09$, $\alpha^\textrm{NL}=3.75$ \cite{3GPP2012TR36.828}.
\iffalse
\begin{table}[!hbt]
  \caption{Simulation Parameters}
  \centering
  \begin{tabular}{|c|c|}
  \hline
  Parameters & Value \\
  \hline
  $P$ & 24 dBm \cite{3GPP2012TR36.828} \\
  \hline
  $N_0$ & -95 dBm \cite{3GPP2012TR36.828}\\
  \hline
  $\gamma_0$ & 0 dB \\
  \hline
  $C$ & 11.95 \cite{mozaffari2016unmanned}\\
  \hline
  $B$ & 0.136 \cite{mozaffari2016unmanned}\\
  \hline
  $h$ & 50m, 100m \\
  \hline
  \multicolumn{2}{|c|} {\quad \quad High-altitude model \cite{fotouhi2017dynamic} and ultra-low-altitude model \cite{3GPP2012TR36.828}\quad \quad} \\
  \hline
  $A^\textrm{L}$ & 10.38 \\
  \hline
  $A^\textrm{NL}$ & 14.54 \\
  \hline
  $\alpha^\textrm{L}$ & 2.09 \\
  \hline
  $\alpha^\textrm{NL}$ & 3.75 \\
  \hline
  \multicolumn{2}{|c|} {Low-altitude model \cite{3GPP2012TR36.828}} \\
  \hline
  $A^\textrm{L}$ & 10.34 \\
  \hline
  $A^\textrm{NL}$ & 13.11 \\
  \hline
  $\alpha^\textrm{L}$ & 2.42 \\
  \hline
  $\alpha^\textrm{NL}$ & 4.28 \\
  \hline
  \end{tabular}
  \label{parameter}
\end{table}
\fi
\subsection{The coverage probability for hovering UAVs}
Fig. \ref{coverage lower 50m} and Fig. \ref{coverage lower 100m} show the comparison of the coverage probability for UAVs hovering at 50m and 100m based on the investigated three models of path loss, i.e., the high-altitude model, the low-altitude model and the ultra-low-altitude model.

It can be seen from Fig. \ref{coverage lower 50m} that with the increase of the UAV density, the coverage probability of the high-altitude model first rises to the peak and then decreases. The optimal UAV density for this model is about 10 BSs/km$^2$. As for the low-altitude model, the performance trend is similar to that of the high-altitude one, with a slightly different optimal density around 6 BSs/km$^2$. The explanations of these phenomena are:
\begin{itemize}
\item For a sparse UAV-BSs density, the distance from associated UAV-BS to UE decreases with the increasing UAV-BS density and the associated UAV-BS is more likely to have a LoS transmission with UE, so the coverage probability grows as the UAV-BS density increases.
\item For a dense UAV-BSs density, although the associated UAV have a higher probability to transmit data via a LoS channel, other UAVs also produce strong interference through LoS paths, thus, the coverage probability decreasing after reaching the optimal point.
\end{itemize}

\begin{figure}
  \centering
  \includegraphics[width=2.8in]{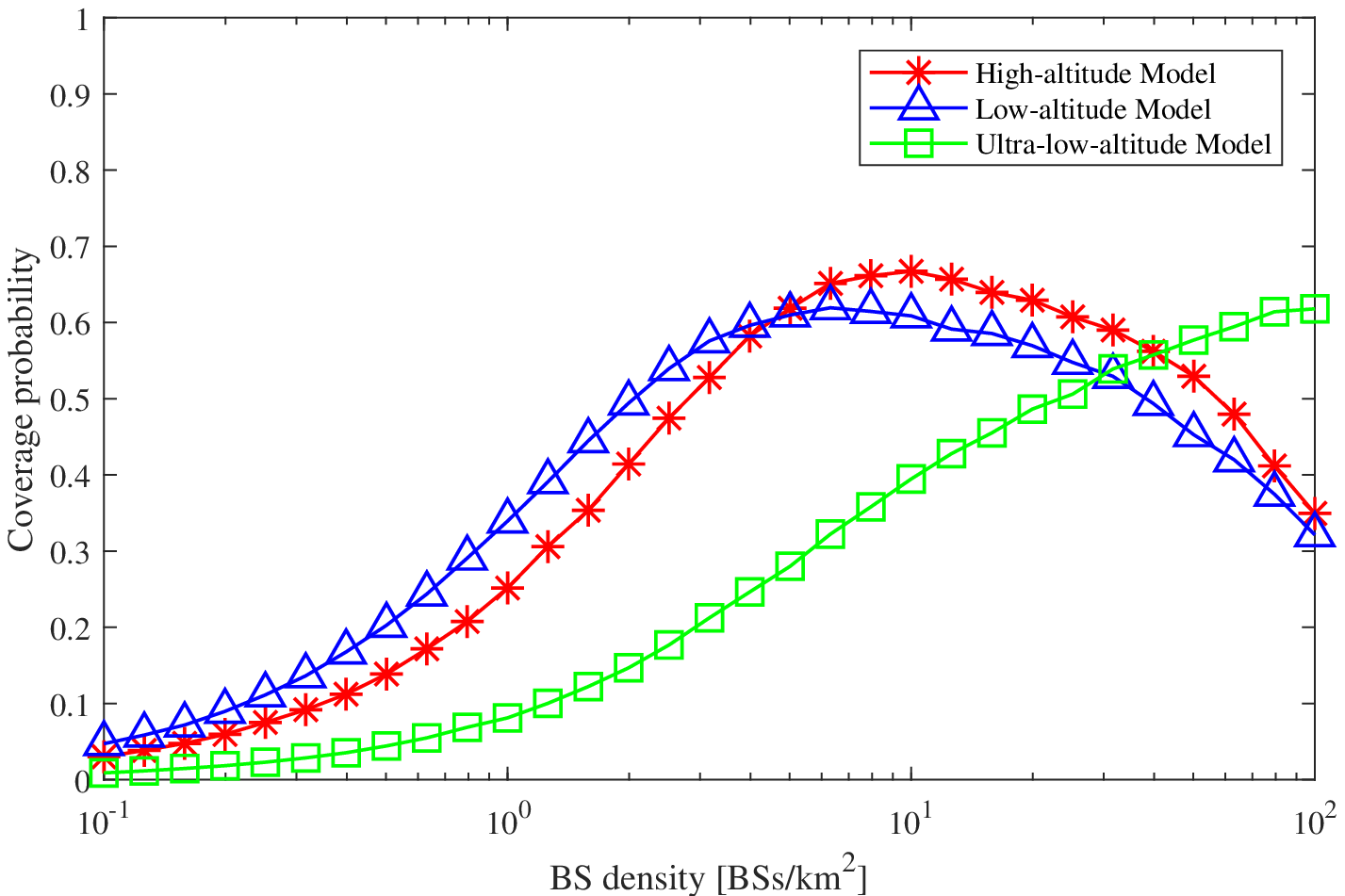}\\
  \caption{Comparison of the coverage probability for hovering UAVs (h=50m)}
  \label{coverage lower 50m}
\end{figure}

\begin{figure}
  \centering
  \includegraphics[width=2.8in]{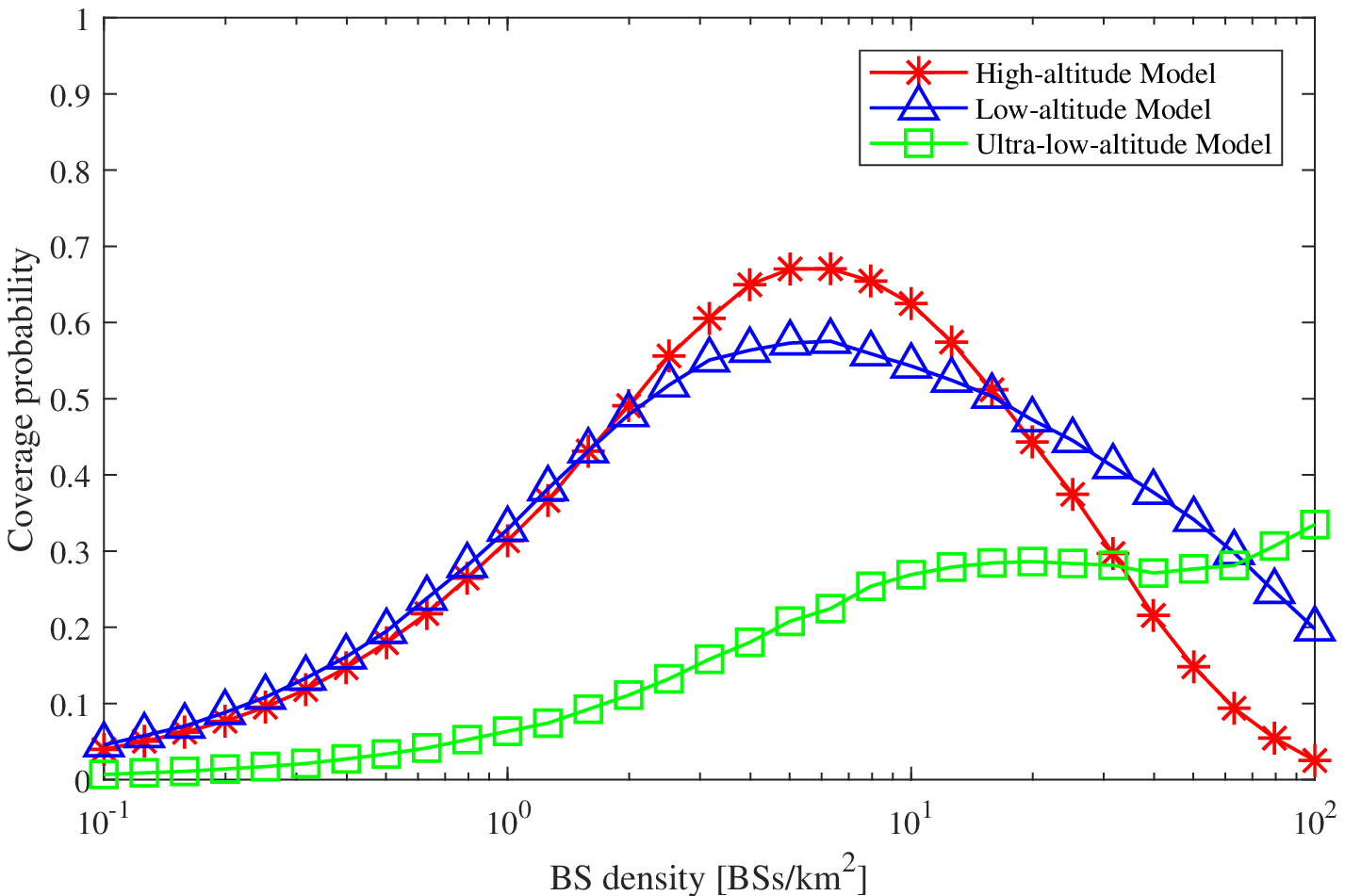}\\
  \caption{Comparison of the coverage probability for hovering UAVs (h=100m)}
  \label{coverage lower 100m}
\end{figure}
For the ultra-low-altitude model, the performance is significantly different with the other two models. The reason is that the ultra-low-altitude model is designed for a scenario where UAVs fly at a relatively low altitude and the transmission distance is quite limited. Furthermore, even when the UAV is hovering over user's location, the probability of having LoS connection is still low since the minimum distance from UAV-BS to UE is the height of UAV. As a result, the ultra-low-altitude model is not suitable for the practical UAV scenario with a height around 50$\sim$100 meters.

In Fig. \ref{coverage lower 100m}, we can see that the performance of the low-altitude model and high-altitude model is very similar when the UAV-BS density is less than 2 BSs/km$^2$. When the density is between 2 BSs/km$^2$ and 20 BSs/km$^2$, the coverage probability of the high-altitude model is higher than that of the low-altitude model, but the low-altitude model performs better than the high-altitude model when density is beyond 20 BSs/km$^2$.
\subsection{The ASE for hovering UAVs}
Fig. \ref{ase low 50m} shows the ASE performance of different path loss models for a height of 50m.
As can be seen from this figure, the ASE of the high-altitude model and the low-altitude model keep growing due to the increasing coverage probability, but the growing rate slows down when the density of UAVs is more than 10 BSs/km$^2$. This is because the declining coverage probability shown in Figs. 3 and 4 outweighs the increase of the UAV density. In this figure we can also find that the ASE for the ultra-low-altitude model differs from the other two. As a result, when the height of UAV is around 50m, the high-altitude model and the low-altitude model are equally good for the performance analysis of the UAV-based network.
\begin{figure}
  \centering
  \includegraphics[width=2.8in]{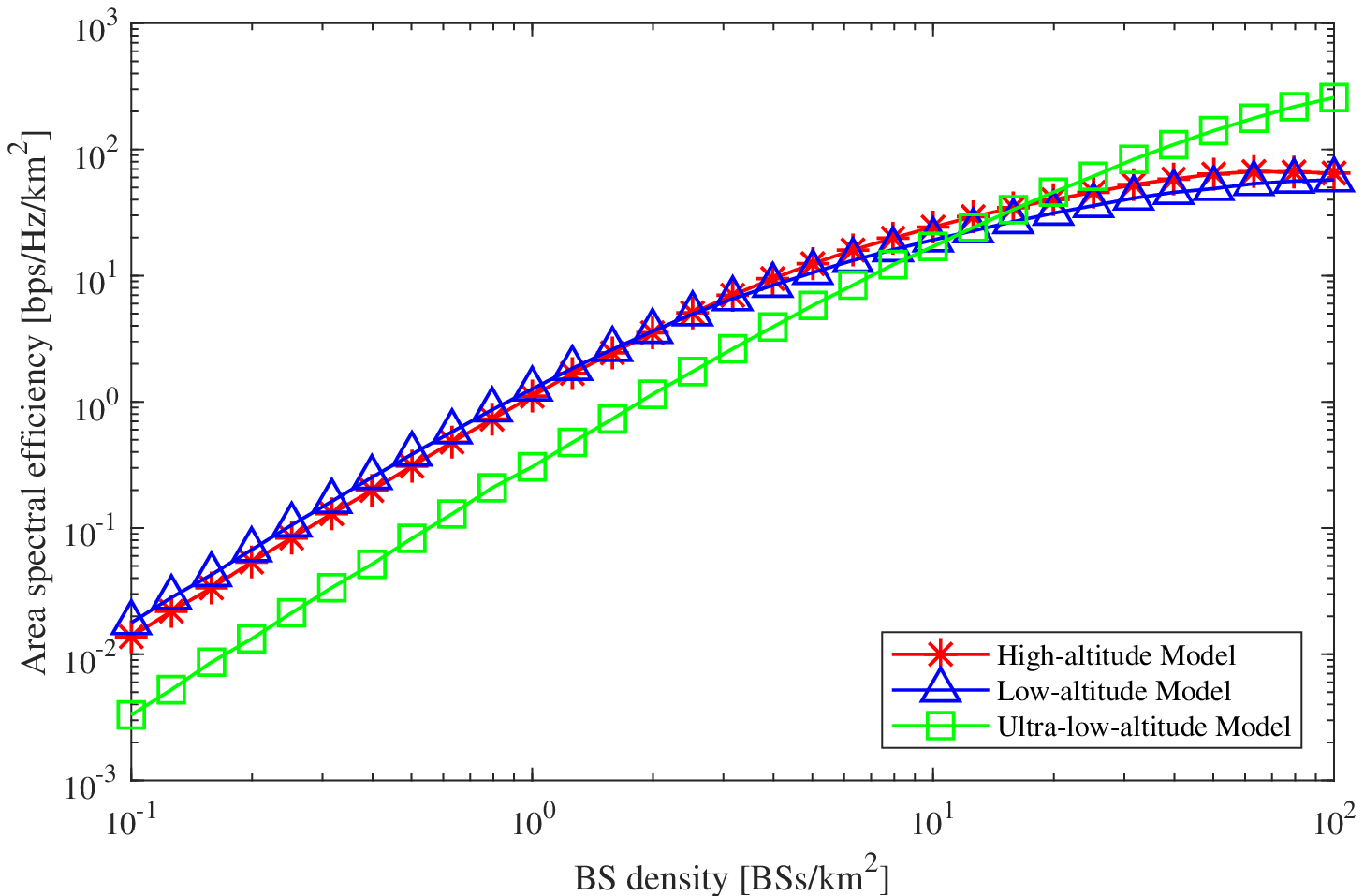}\\
  \caption{Comparison of the ASE for hovering UAVs (h=50m)}
  \label{ase low 50m}
\end{figure}
\begin{figure}
\centering
  \includegraphics[width=2.8in]{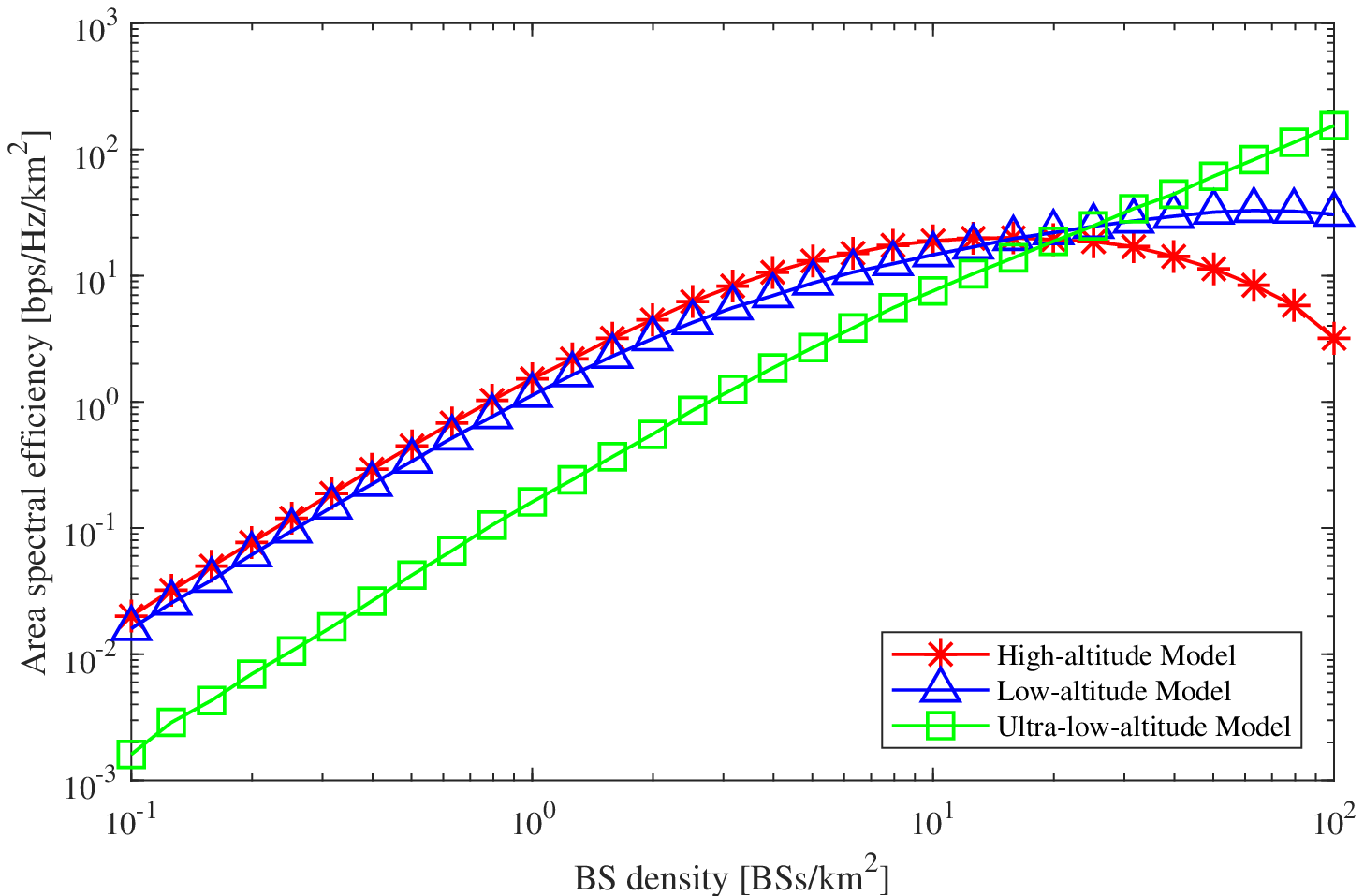}\\
  \caption{Comparison of the ASE for hovering UAVs (h=100m)}
  \label{ase low 100m}
\end{figure}

Fig. \ref{ase low 100m} shows the comparison of the ASE in different models when the height of UAV is 100m. When the UAV-BS density is lower than 20 BSs/km$^2$, the ASE of the high-altitude model and that of the low-altitude model leave the similar trail. However, after reaching the density of 10 BSs/km$^2$, their ASE performance diverges. The drop of the ASE for the high-altitude model indicates that the deceasing coverage dominates the ASE performance compared with the growing UAV density. Considering that the low-altitude model was developed for a height around 32 meters, it might not be suitable for the UAVs flying at 100 meters studied here. Hence, the high-altitude model might be more appropriate here. However, we may need to conduct real-life channel measurement to confirm this conjecture.

\subsection{The performance for Teleporting UAVs}
It can be seen from the previous simulation that when UAVs fly at the height of 50m, the coverage probability and the ASE performance of the high-altitude model and the low-altitude model are very similar. However, when the height of UAV is at 100m, the performance of these two models deviate in dense networks. To verify whether these two models are still equally good for teleporting UAVs at 50m, we investigate and compare their coverage probability and ASE performance in this subsection.
%As a result, we further do the simulation of teleporting UAVs at 50m to verify the availability of the high-altitude model and low-altitude model.
\begin{figure}
  \centering
  \includegraphics[width=2.8in]{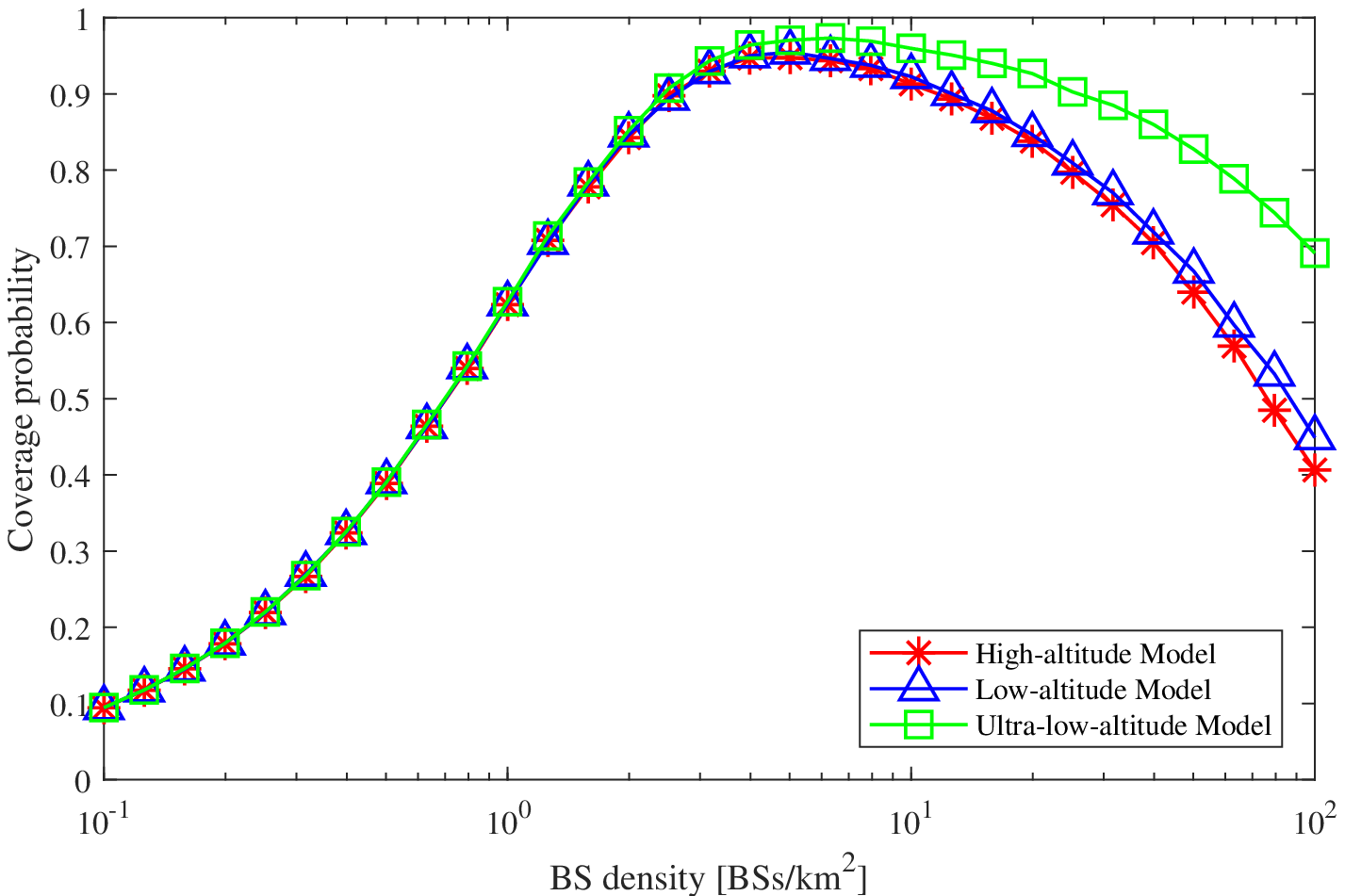}\\
  \caption{Comparison of the coverage probability for teleporting UAVs (h=50m)}
  \label{coverage upper 50m}
\end{figure}

\begin{figure}
\centering
  \includegraphics[width=2.8in]{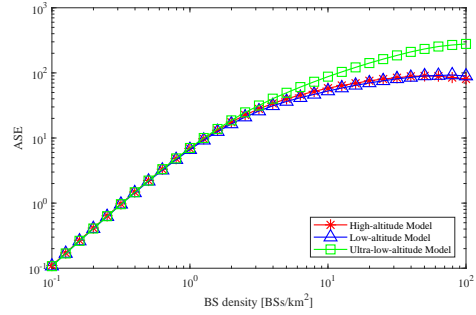}\\
  \caption{Comparison of the ASE for teleporting UAVs (h=50m)}
  \label{ase up 50m}
\end{figure}

From Fig. \ref{coverage upper 50m} and Fig. \ref{ase up 50m}, we can see that the high-altitude model and the low-altitude model generate similar results.  In Fig. \ref{coverage upper 50m}, we can see that the optimal UAV density for these two models can be found at around 6 BSs/km$^2$. Fig. \ref{ase up 50m} shows that the ASE of both increases linearly at first, and then grows slowly.

%Fig. \ref{coverage upper 50m} shows that when UAVs fly at the the altitude of 50m,if they can respond to users' service request and hover at the user's location instantaneously,the coverage probability for the high-altitude model and the low-altitude model are highly similar, and . In Fig. \ref{ase up 50m}, we can see that the ASE for the high-altitude model and the low-altitude model are relatively similar, both of them increases linearly at first, and then gradually become steady. It also indicates that the high-altitude model and the low-altitude model are still equally good in the application of teleporting UAVs at the height of 50m.
\subsection{Comparison of the Upper and Lower bounds of Performance}
From Fig. \ref{coverage lower 50m} and Fig. \ref{ase low 50m}, we can see that the high-altitude model and the low-altitude model are equally good for network performance analysis. Hence, we choose the high-altitude model to show the difference between the upper bound of ASE and the lower bound of ASE when the UAVs fly at the height of 50m. Such comparison is displayed in Fig. \ref{ase compare}. It can be seen that when the density is lower than 10 BSs/km$^2$, the gap between the upper bound and the lower bound is large, which shows great promise for optimization of UAV mobility in UAV-enabled networks. However, as UAV density increases, the ASE gain due to the UAV mobility becomes marginal, e.g., at a UAV density of 100 UAV-BSs/km$^2$.
\begin{figure}
\centering
  \includegraphics[width=2.8in]{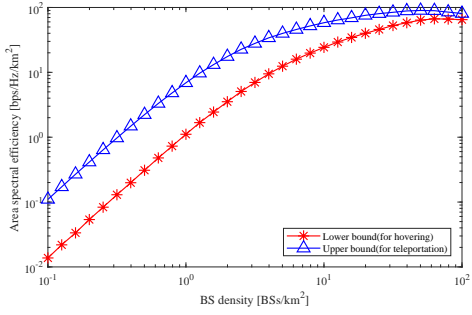}\\
  \caption{Upper bound and lower bound of the ASE for the high-altitude model when h=50m}
  \label{ase compare}
\end{figure}
\section{Conclusion}
In this paper, we have studied the performance of UAV-enabled wireless networks. In order to identify the proper path loss models for UAVs flying at practical heights, such as 50m and 100m, we first analyze the performance when adopting the conventional high-altitude model based on the elevation angle. Then we further investigate the coverage probability and the ASE by using path loss models which have been widely applied to terrestrial communications, including the low-altitude model and the ultra-low-altitude model. From simulation results, we find that performance for networks with high-altitude model and the low-altitude model are equally good when UAVs fly at the height of 50m, while the performance trend of the ultra-low-altitude model is quite different. We also found that the number of the UAVs should be optimized for the benefits of the networks, which sheds new light on the design of the future UAV-enabled networks.

\begin{appendices}
\section{Proof of Theorem \ref{1}}
When transmission between UAV-BS and UE is LoS, two conditions should be satisfied: 1) The distance between UE and UAV-BS is $r$, and there is no UAV-BS of LoS path within $r$.
%, denoted as event $B^\textrm{L}$.
2) There is no UAV-BS of NLoS path within ${r_1}$.
%, denoted as event $C^\textrm{NL}$.
Based on these conditions, $f^\textrm{L}(r)$ is computed as \cite{ding2016performance}
\begin{equation}
 \begin{split}
f^\textrm{L}(r) = &\exp \left( { - \int_0^{{r_1}} {\left( {1 - {{\Pr }^\textrm{L}}\left( {u} \right)} \right)2\pi u\lambda du} } \right)\\
 &\times \exp \left( { - \int_0^r {{{\Pr }^\textrm{L}}\left( {u} \right)2\pi u\lambda du} } \right)\\
 &\times {\Pr} ^\textrm{L}\left( {r} \right) \times 2\pi r\lambda,
 \end{split}
\end{equation}

When transmission between UAV-BS and UE is NLoS, two conditions should be satisfied: 1) The distance between UE and UAV-BS is $r$, and there is no UAV-BS of NLoS path within $r$.
%, denoted as event $D^\textrm{NL}$.
2) There is no UAV-BS of LoS path within ${r_2}$.
%,denoted as event $E^\textrm{L}$.
So $f^\textrm{NL}(r)$ can be derived as
\begin{equation}
 \begin{split}
f^{\textrm{NL}}(r) = &\exp \left( { - \int_0^{{r_2}} {{{\Pr }^\textrm{L}}\left( {u} \right)2\pi u\lambda du} } \right)\\
 &\times \exp \left( { - \int_0^r {\left( {1 - {{\Pr }^\textrm{L}}\left( {u} \right)} \right)2\pi u\lambda du} } \right)\\
 &\times \left( {1 - {\Pr}^\textrm{L}\left( {r} \right)} \right) \times 2\pi r\lambda.
 \end{split}
\end{equation}
\iffalse the CDF of associated UAV-BS located at $r$ from UE is $1-\exp \left( -{\int_0^r {{{\Pr }^{\rm{L}}}(u)2\pi u\lambda du} } \right)$, and PDF is
\begin{equation}
\begin{split}
{f^{\rm{L}}}(r) = \exp \left( { - \int_0^r {{{\Pr }^{\rm{L}}}(u)2\pi u\lambda du} } \right){{\Pr} ^{\rm{L}}}(r)2\pi r\lambda
\end{split}
\end{equation}
Considering the probability of event $C^\textrm{NL}$ conditioned on associated UAV-BS located at $r$ km from UE, it can be derived as
\begin{equation}
\begin{split}
\Pr \left[ {{C^{{\rm{NL}}}}\left| {{R^{\rm{L}}} = r} \right.} \right] = \exp \left( { - \int_0^{{r_1}} {\left( {1 - {{\Pr }^{\rm{L}}}(u)} \right)} 2\pi u\lambda du} \right)
\end{split}
\end{equation}
\fi
 ${\Pr \left[ {{\rm{SINR > }}\gamma \left| r \right.} \right]}$ conditioned on $r$ can be expressed as
 \iffalse
\begin{equation}
\begin{split}
\Pr \left[ {\frac{{P{\zeta ^{\rm{L}}}(r)g}}{{{I_r} + {N_0}}} > \gamma } \right] = &{\mathbb{E}_{\left[ {{I_r}} \right]}}\left\{ {g > \frac{{\gamma \left( {{I_r} + {N_0}} \right)}}{{P{\zeta ^{\rm{L}}}(r)}}} \right\}\\
 = &{\mathbb{E}_{\left[ {{I_r}} \right]}}\left\{ {{{\bar F}_G}\left( {\frac{{\gamma \left( {{I_r} + {N_0}} \right)}}{{P{\zeta ^{\rm{L}}}(r)}}} \right)} \right\},
\end{split}
\end{equation}
where ${{\bar F}_G}$ is the CCDF of $g$ which has exponential distribution. So ${{\bar F}_G} = \exp (-g)$ and ${\Pr \left[ {{\rm{SINR > }}\gamma \left| r \right.} \right]}$ can be further expressed as
\fi
\begin{equation}
\begin{split}
&\Pr \left[ {\frac{{P{\zeta ^{\rm{L}}}(r)g}}{{{I_r} + {N_0}}} > \gamma } \right]
=  {\mathbb{E}_{\left[ {{I_r}} \right]}}\left\{ {\exp \left( { - \frac{{\gamma \left( {{I_r} + {N_0}} \right)}}{{P{\zeta ^{\rm{L}}}(r)}}} \right)} \right\}\\
=& \exp \left( { - \frac{{\gamma {N_0}}}{{P{\zeta ^{\rm{L}}}(r)}}} \right){\mathbb{E}_{\left[ {{I_r}} \right]}}\left\{ {\exp \left( { - \frac{{\gamma {I_r}}}{{P{\zeta ^{\rm{L}}}(r)}}} \right)} \right\}\\
=& \exp \left( { - \frac{{\gamma {N_0}}}{{P{\zeta ^{\rm{L}}}(r)}}} \right){\mathcal{L}_{{I_r}}}\left( {\frac{\gamma }{{P{\zeta ^{\rm{L}}}(r)}}} \right).
\end{split}
\end{equation}
Then $\Pr \left[ {\frac{{P{\zeta ^{\rm{NL}}}(r)g}}{{{I_r} + {N_0}}} > \gamma } \right]$ can be derived in the similar way.
\iffalse
\begin{equation}
\begin{split}
&\Pr \left[ {\frac{{P{\zeta ^{\rm{NL}}}(r)g}}{{{I_r} + {N_0}}} > \gamma } \right] \\
=  & \exp \left( { - \frac{{\gamma {N_0}}}{{P{\zeta ^{\rm{NL}}}(r)}}} \right){\mathbb{E}_{\left[ {{I_r}} \right]}}\left\{ {\exp \left( { - \frac{{\gamma {I_r}}}{{P{\zeta ^{\rm{NL}}}(r)}}} \right)} \right\}\\
=& \exp \left( { - \frac{{\gamma {N_0}}}{{P{\zeta ^{\rm{NL}}}(r)}}} \right){\mathcal{L}_{{I_r}}}\left( {\frac{\gamma }{{P{\zeta ^{\rm{NL}}}(r)}}} \right).
\end{split}
\end{equation}
\fi
\end{appendices}

\bibliographystyle{IEEEtran}
\bibliography{UAV}

\begin{thebibliography}{10}
\providecommand{\url}[1]{#1}
\csname url@samestyle\endcsname
\providecommand{\newblock}{\relax}
\providecommand{\bibinfo}[2]{#2}
\providecommand{\BIBentrySTDinterwordspacing}{\spaceskip=0pt\relax}
\providecommand{\BIBentryALTinterwordstretchfactor}{4}
\providecommand{\BIBentryALTinterwordspacing}{\spaceskip=\fontdimen2\font plus
\BIBentryALTinterwordstretchfactor\fontdimen3\font minus
  \fontdimen4\font\relax}
\providecommand{\BIBforeignlanguage}[2]{{%
\expandafter\ifx\csname l@#1\endcsname\relax
\typeout{** WARNING: IEEEtran.bst: No hyphenation pattern has been}%
\typeout{** loaded for the language `#1'. Using the pattern for}%
\typeout{** the default language instead.}%
\else
\language=\csname l@#1\endcsname
\fi
#2}}
\providecommand{\BIBdecl}{\relax}
\BIBdecl

\bibitem{zhang2017spectrum}
C.~Zhang and W.~Zhang, ``Spectrum sharing for drone networks,'' \emph{IEEE
  Journal on Selected Areas in Communications}, vol.~35, no.~1, pp. 136--144,
  2017.

\bibitem{ono2016wireless}
F.~Ono, H.~Ochiai, and R.~Miura, ``A wireless relay network based on unmanned
  aircraft system with rate optimization,'' \emph{IEEE Transactions on Wireless
  Communications}, vol.~15, no.~11, pp. 7699--7708, 2016.

\bibitem{gruber2016role}
M.~Gruber, ``Role of altitude when exploring optimal placement of uav access
  points,'' in \emph{Wireless Communications and Networking Conference (WCNC),
  2016 IEEE}.\hskip 1em plus 0.5em minus 0.4em\relax IEEE, 2016, pp. 1--5.

\bibitem{galkin2016deployment}
B.~Galkin, J.~Kibilda, and L.~A. DaSilva, ``Deployment of uav-mounted access
  points according to spatial user locations in two-tier cellular networks,''
  in \emph{Wireless Days (WD), 2016}.\hskip 1em plus 0.5em minus 0.4em\relax
  IEEE, 2016, pp. 1--6.

\bibitem{mozaffari2015drone}
M.~Mozaffari, W.~Saad, M.~Bennis, and M.~Debbah, ``Drone small cells in the
  clouds: Design, deployment and performance analysis,'' in \emph{Global
  Communications Conference (GLOBECOM), 2015 IEEE}.\hskip 1em plus 0.5em minus
  0.4em\relax IEEE, 2015, pp. 1--6.

\bibitem{alzenad20173d}
M.~Alzenad, A.~El-Keyi, F.~Lagum, and H.~Yanikomeroglu, ``3d placement of an
  unmanned aerial vehicle base station (uav-bs) for energy-efficient maximal
  coverage,'' \emph{IEEE Wireless Communications Letters}, 2017.

\bibitem{mozaffari2016unmanned}
M.~Mozaffari, W.~Saad, M.~Bennis, and M.~Debbah, ``Unmanned aerial vehicle with
  underlaid device-to-device communications: Performance and tradeoffs,''
  \emph{IEEE Transactions on Wireless Communications}, vol.~15, no.~6, pp.
  3949--3963, 2016.

\bibitem{bor2016efficient}
R.~I. Bor-Yaliniz, A.~El-Keyi, and H.~Yanikomeroglu, ``Efficient 3-d placement
  of an aerial base station in next generation cellular networks,'' in
  \emph{Communications (ICC), 2016 IEEE International Conference on}.\hskip 1em
  plus 0.5em minus 0.4em\relax IEEE, 2016, pp. 1--5.

\bibitem{ding2016please}
M.~Ding and D.~L. P{\'e}rez, ``Please lower small cell antenna heights in 5g,''
  in \emph{Global Communications Conference (GLOBECOM), 2016 IEEE}.\hskip 1em
  plus 0.5em minus 0.4em\relax IEEE, 2016, pp. 1--6.

\bibitem{ding2016performance}
M.~Ding, P.~Wang, D.~L{\'o}pez-P{\'e}rez, G.~Mao, and Z.~Lin, ``Performance
  impact of los and nlos transmissions in dense cellular networks,'' \emph{IEEE
  Transactions on Wireless Communications}, vol.~15, no.~3, pp. 2365--2380,
  2016.

\bibitem{3GPP2012TR36.828}
3GPP, ``Tr 36.828 (v11.0.0): Further enhancements to lte time division duplex
  (tdd) for downlink-uplink (dl-ul) interference management and traffic
  adaptation,'' 2016.

\bibitem{ding2017new}
M.~Ding, D.~L. P{\'e}rez, and G.~Mao, ``A new capacity scaling law in
  ultra-dense networks,'' \emph{arXiv preprint arXiv:1704.00399}, 2017.

\bibitem{fotouhi2017dynamic}
A.~Fotouhi, M.~Ding, and M.~Hassan, ``Dynamic base station repositioning to
  improve spectral efficiency of drone small cells,'' \emph{arXiv preprint
  arXiv:1704.01244}, 2017.

\end{thebibliography}
\end{document}